\documentclass[12pt]{article}

\usepackage{amsmath,amsthm,amssymb,amsfonts}
\usepackage{stmaryrd}
\usepackage{graphicx}
\usepackage{enumerate}
\usepackage{bbm}
\usepackage{tikz}

\newtheorem{theorem}{Theorem}
\newtheorem{lemma}[theorem]{Lemma}

\newtheorem{corollary}[theorem]{Corollary}

\newcommand{\var}{\mathrm{var}}
\newcommand{\ex}[1]{E\left[ #1 \right]}
\newcommand{\exi}[1]{E_{i}\left[ #1 \right]}

\title{On the maximal L1 influence of \\ real-valued boolean functions}
\author{Andrew J. Young and Henry D. Pfister}
\date{}

\begin{document}

\maketitle

\begin{abstract}
We show that any sequence of well-behaved (e.g. bounded and non-constant) real-valued functions of $n$ boolean variables $\{f_{n}\}$ admits a sequence of coordinates whose $L^1$ influence under the $p$-biased distribution, for any $p\in(0,1)$, is $\Omega(\var(f_n) \frac{\ln n}{n})$.
\end{abstract}

\section{Introduction}

The celebrated KKL result of Kahn, Kalai and Linial~\cite{KKL1988} shows that any boolean-valued function of $n$ boolean variables has a variable whose influence is $\Omega \big( \var(f) \frac{\ln(n)}{n} \big)$ which is a factor $\ln(n)$ larger than predicted.
The definition of influence in this result is the classic one and has many equivalent formulations.
The definition we use is introduced in the next paragraph and discussed in Section~\ref{sec:background}.
Some related work \cite{Bourgain1992,Hatami2009,Keller2011,Cordero-Erausquin2012,Keller2012,ODonnell2013} has sought to generalize the class of applicable functions, in particular the domain, using a variety of definitions for influence, e.g. geometric.

We provide a similar result for the $L^1$ influence of real-valued boolean functions (Theorem \ref{theorem_real_kkl}), where the \emph{$L^p$ influence} of variable $i$ equals the $L^p$ norm to the $p$-th power of the difference between the function and its average over the $i$-th coordinate.
This has implications for sharp thresholds (e.g. see~\cite{Talagrand1994,Friedgut1996}) that follow from Rossignol's generalization of a lemma~\cite{Rossignol2006} due to Margulis and Russo~\cite{Margulis1974,Russo1982}. 
The proof is based on a variation of the hypercontractivity theorem for $p$-biased measures studied by Talagrand \cite{Talagrand1994} (Corollary \ref{corollary_p-hypercontractivity}).

More recently, Kelman et al. \cite{Kelman2021} analyze $L^1$ influences to provide variations of several well known theorems.
Their results hold for bounded functions and $p=1/2$.
In contrast, our results apply to a wider class of functions and any $0<p<1$.
Moreover, our proof technique is quite different.

The final result is of an auxiliary nature and, as such, some details are omitted.
It concerns a converse, up to a small multiplicative factor, to Theorem \ref{theorem_real_kkl} achieved by a well known sequence of boolean functions, the tribes functions of Ben-Or and Linial \cite{Ben-Or1990}.
While this result is well known, we include if for completeness and to demonstrate the tightness of the constant in Theorem \ref{theorem_real_kkl}.

\section{Main result}

We state our main Theorem and some direct consequences.
Some standard definitions are deferred to the subsequent sections.

\begin{theorem}
\label{theorem_real_kkl}
Let $\mu$ be the $p$-biased measure, $f_{n} : \{-1,1\}^{n} \rightarrow \mathbb{R}$ and $f_{n}^{(i)} = f_{n} - \exi{f_{n}}$.
If $\var(f_{n})$ is strictly positive and $o(n^{\varepsilon})$, for all $\varepsilon > 0$, then 
\begin{equation*}
\liminf_{n \rightarrow \infty} \frac{ \max_{i} \lVert f_{n}^{(i)} \rVert_{1} }{ \var(f_{n}) \frac{\ln n}{n} } \ge \frac{C_{0}}{M_{0}},
 \end{equation*}
where
\begin{equation*}
M_{0} = \limsup_{n \rightarrow \infty} \max_{i : f_{n}^{(i)} \ne 0} \frac{ \lVert f_{n}^{(i)} \rVert_{2}^{2}}{ \lVert f_{n}^{(i)} \rVert_{1}}
\qquad
C_{0} = \sup_{ \alpha > 0} \frac{\tanh\left( \frac{\alpha}{2} \right) }{ \alpha - \ln \rho_{2}(\alpha)^{2} },
\end{equation*}
$\rho_{2}(\alpha) = \rho \left( e^{\alpha} + 1 \right)$ and $\rho$ is any of the smoothing parameters in Theorem \ref{theorem_p-hypercontractivity}.
\end{theorem}

If $M_{0}$ happens to be $0$, then the constant $C_{0} / M_0$ is interpreted as infinity and the RHS of Theorem \ref{theorem_real_kkl} can be taken to be any desired nonnegative constant.

For $p=1/2$, it is known that the constant in the original KKL theorem can be improved to $1/2$ (see e.g. \cite{odonnell14} Exercise 9.30).
This is a direct corollary because the $L^1$ and $L^2$ influences coincide for boolean functions giving $M_{0} = 1$ and $C_{0} = \frac{1}{2}$ when $p=\frac{1}{2}$.
Moreover, by letting $\alpha = 1$ and applying H\"{o}lder's inequality, we see that one implication of Theorem \ref{theorem_real_kkl} is
\begin{equation}
\label{equation_real_kkl_bounded}
\liminf_{n \rightarrow \infty} \frac{ \max_{i} \lVert f_{n}^{(i)} \rVert_{1} }{ \var(f_{n}) \frac{\ln n}{n} }
\ge \frac{9}{20} \frac{1}{\sup_{n} \lVert f_{n} \rVert_{\infty}} \frac{1}{1 + \left \lvert \ln \frac{p}{1-p} \right \rvert},
\end{equation}
where $\rho$ has the form of item $(iii)$ in Theorem \ref{theorem_p-hypercontractivity}.

\section{Fourier analysis on the $p$-biased hypercube}
\label{sec:background}

The domain of most functions is the Cartesian product of $\{-1,1\}$, and, typically, we assign $-1$ weight $1-p$ and $1$ weight $p$.
Such functions whose range is the real numbers will be referred to as \emph{real-valued boolean functions}.
For any $n$, the functions $\tau_{i}^{\pm} : \{-1,1\}^{n} \rightarrow \{-1,1\}^{n}$ fix the $i$-th coordinate to be $1$ for $+$ and $-1$ for $-$ and operate as the identity on all remaining coordinates.
Given a measure $\mu$, we use $\int \cdot \, d \mu$ and $\ex{\cdot}$ interchangeably.
Moreover, $\exi{\cdot}$ is integration over only the $i$-th coordinate, the $L^p$ norms are defined in the usual way (i.e. $\lVert f \rVert_{q}^{q} = \int \lvert f \rvert^{q} \, d \mu$) and the $L^q$ influence of the $i$-th coordinate is $\lVert f - \exi{f} \rVert_{q}^{q}$.

Let $\mu_{i}$ be the measure on $\{-1,1\}$ with $\exi{x_{i}} = 2 p_{i} - 1$ for some $0 < p_{i} < 1$ and $\mu = \mu_{1} \otimes \cdots \otimes \mu_{n}$.
The canonical orthonormal basis is
\begin{equation*}
\chi_{S} = \prod_{i \in S} \chi_{i},
\end{equation*}
where
\begin{equation*}
\chi_{i}(x) = \frac{1}{\sigma_{i}} \left( x_{i} - \exi{x_{i}} \right)
\end{equation*}
and
\begin{equation*}
\sigma_{i}^{2} = \exi{ \left( x_{i} - \exi{x_{i}} \right)^{2} }
\end{equation*}
is the variance.
More explicitly,
\begin{equation*}
\chi_{i}(x)
=
\begin{cases}
-\sqrt{\frac{p_{i}}{1-p_{i}}} & \text{if } x_{i} = -1 ,\\
\sqrt{\frac{1-p_{i}}{p_{i}}} & \text{if } x_{i} = 1.
\end{cases}
\end{equation*}

Every real-valued boolean function has a Fourier expansion
\begin{equation*}
f = \sum_{S \subseteq [n] } \hat{f}(S) \chi_{S},
\end{equation*}
where $[n]$ is $\{1,2, \ldots, n \}$ and
\begin{equation*}
\hat{f}(S) = \left< f , \chi_{S} \right>
= \sum_{x \in \{-1,1\}^{n}} f(x) \chi_{S}(x) \mu(x)
\end{equation*}
are the Fourier coefficients.
By Parseval's theorem the variance of such a function is the sum of all its squared Fourier coefficients not indexed by the empty set.
The smoothing operator $T_{\delta}$ acts on real-valued boolean functions as
\begin{equation*}
T_{\delta} f := \sum_{S \subseteq [n]} \delta^{\lvert S \rvert} \hat{f}(S) \chi_{S}.
\end{equation*}

If $p_{i} = p$ for all $i$, then $\mu$ is called the \emph{$p$-biased measure}.

The following Theorem and its immediate Corollary provide the necessary $p$-biased variation of the hypercontractivity theorem for our purposes.

\begin{theorem}
\label{theorem_p-hypercontractivity}
Let $\mu$ be the $p$-biased measure and $f : \{-1,1\}^{n} \rightarrow \mathbb{R}$, then for all $q \ge 2$
\begin{equation*}
\lVert T_{\gamma} f \rVert_{q} \le \lVert f \rVert_{2},
\end{equation*}
where $\gamma = \frac{1}{\sqrt{q-1}} \rho(q)$ and $\rho(q)$ is any one of the following:
\begin{enumerate}[i)]
\item
\begin{equation*}
\lambda^{\frac{1}{2}-\frac{1}{q}};
\end{equation*}
\item
\begin{equation*}
\sqrt{q-1} \sqrt{ \frac{ \sinh\left( - \frac{1}{q} \ln \left( \frac{ \lambda}{1- \lambda} \right) \right) }{ \sinh\left( - \left( 1 - \frac{1}{q} \right) \ln \left( \frac{ \lambda}{1- \lambda} \right) \right) } };
\end{equation*}
\item
\begin{equation*}
\sqrt{\frac{\lambda}{1-\lambda}},
\end{equation*}
where $\lambda = \min\{p,1-p\}$.
\end{enumerate}
\end{theorem}

\begin{proof}
$(i)$ See \cite[Chapter 10]{odonnell14}.
$(ii)$ See \cite{Latala2000}.
$(iii)$
Suppose $q = \infty$.
Then
\begin{equation*}
\lVert T_{0} f \rVert_{\infty}
= \lvert \ex{f} \rvert
\le \lVert f \rVert_{1}
\le \lVert f \rVert_{2}.
\end{equation*}
Suppose $q$ is finite.
Let
\begin{equation*}
\nu_{S}(x)
= \rho^{\lvert S \rvert} \chi_{S}(x)
\end{equation*}
and $m$ be the uniform measure on $\{-1,1\}^{n}$.
Then, Fubini's theorem implies that
\begin{align*}
\lVert T_{ \delta} f \rVert_{q}^{q}
&= \int \lVert T_{ \delta} f \rVert_{q}^{q} \, d m(y) \\
&= \int \int \left| \sum_{S} \left( \frac{ \rho }{\sqrt{q-1}} \right)^{\lvert S \rvert} \hat{f}(S) \chi_{S}(x) \right|^{q} \, d \mu(x) \, d m(y) \\
&= \int \int \left| \sum_{S} \left( \frac{1}{\sqrt{q-1}} \right)^{\lvert S \rvert} \hat{f}(S) \nu_{S}(x) y^{S} \right|^{q} \, d \mu(x) \, d m(y) \\
&= \int \int \left| \sum_{S} \left( \frac{1}{\sqrt{q-1}} \right)^{\lvert S \rvert} \hat{f}(S) \nu_{S}(x) y^{S} \right|^{q} \, d m(y) \, d \mu(x).
\end{align*}
By standard hypercontractivity, for the uniform measure applied to the function with uniform Fourier coefficients $\hat{f}(S) \nu_{S}(x)$, \cite{Beckner1975} and Parseval's theorem
\begin{align*}
\int \left| \sum_{S} \left( \frac{1}{\sqrt{q-1}} \right)^{\lvert S \rvert} \hat{f}(S) \nu_{S}(x) y^{S} \right|^{q} \, d m(y)
&\le \left( \sum_{ S} \hat{f}(S)^{2} \nu_{S}(x)^{2} \right)^{\frac{q}{2}} \\
&\le \left( \sum_{ S} \hat{f}(S)^{2} \right)^{\frac{q}{2}} \\
&= \lVert f \rVert_{2}^{q},
\end{align*}
where, by definition of $\rho$,
\begin{equation*}
\lvert \nu_{S}(x) \rvert
= \rho^{\lvert S \rvert} \prod_{i \in S} \lvert \chi_{i}(x) \rvert
\le \rho^{\lvert S \rvert} \max \left\{ \sqrt{ \frac{p}{1-p}}, \sqrt{ \frac{1-p}{p} } \right\}^{\lvert S \rvert}
= 1.
\end{equation*}
\end{proof}

By \cite{Latala2000}, item $(ii)$ is optimal.

\begin{corollary}
\label{corollary_p-hypercontractivity}
Let $\mu$ be the $p$-biased measure and $f : \{-1,1\}^{n} \rightarrow \mathbb{R}$, then for all $0 \le \delta \le 1$
\begin{equation*}
\lVert T_{ \rho_{1}(\delta) \delta} f \rVert_{2} \le \lVert f \rVert_{1 + \delta^{2}},
\end{equation*}
where $\rho_{1}(\delta) = \rho \left( \frac{1}{\delta^{2}} + 1 \right)$.
\end{corollary}

\begin{proof}
Since $\delta = \frac{1}{\sqrt{q-1}}$ for some $q \ge 2$, we choose $q$ to satisfy this equality.
Let $\gamma = \rho(q) \delta =\rho \left( \frac{1}{\delta^{2}} + 1 \right) \delta$.
Then by H\"{o}lder's inequality
\begin{align*}
\lVert T_{ \gamma} f \rVert_{2}^{2} 
&= \left < T_{\gamma^{2}} f, f \right> \\ 
&\le \lVert T_{\gamma^{2}} f \rVert_{q} \lVert f \rVert_{1 + \frac{1}{q-1}} \\
&= \lVert T_{\gamma} \left( T_{\gamma}  f \right) \rVert_{q}\lVert f \rVert_{1 + \frac{1}{q-1}} \\
&\le \lVert T_{\gamma}f \rVert_{2}\lVert f \rVert_{1 + \frac{1}{q-1}}.
\end{align*}
\end{proof}

\begin{lemma}
\label{lemma_expectation_spectrum}
Let $\mu$ be the $p$-biased measure, $f : \{-1,1\}^{n} \rightarrow \mathbb{R}$ and $f_{i} = f - \exi{f}$.
Then, for all $i$, we have
\begin{enumerate}[i)]
\item
\begin{equation*}
\hat{f_{i}}(S)
= 
\begin{cases}
\hat{f}(S) & \text{if } S \ni i, \\
0 & \text{otherwise};
\end{cases}
\end{equation*}
\item
\begin{equation*}
f_{i} = \frac{\sigma}{2} \left( f \circ \tau_{i}^{+} - f\circ \tau_{i}^{-} \right) \chi_{i}.
\end{equation*}
\end{enumerate}
\end{lemma}

\begin{proof}
$(i)$
\begin{equation*}
\exi{\chi_{S}}
= 
\begin{cases}
\chi_{S} & S \not \ni i \\
0 & \text{else}
\end{cases}.
\end{equation*}
$(ii)$
\begin{equation*}
\chi_{S} \circ \tau_{i}^{+} - \chi_{S} \circ \tau_{i}^{-}
= 
\begin{cases}
\frac{2}{\sigma} \chi_{S \backslash \{i\} } & S \ni i \\
0 & \text{else}
\end{cases}.
\end{equation*}
\end{proof}

\section{Proof of Theorem \ref{theorem_real_kkl}}

The function $\rho(q)$ will be used often and with varying parameterization.
In particular, when $\delta$ is used, following the notation in Corollary \ref{corollary_p-hypercontractivity}, $\rho_{1}(\delta) = \rho \left( \frac{1}{\delta^{2}} + 1 \right)$ and, when $\alpha,\alpha_{0}$ are used, as in Theorem \ref{theorem_real_kkl}, $\rho_{2}(\alpha) = \rho \left( e^{\alpha} + 1 \right)$.

If $M_{0} = \infty$, there is nothing to prove, i.e. the LHS is nonnegative.
Suppose $M_{0}$ is finite.
Fix $\alpha_{0} > 0$.
Relabel $\{f_{n}\}$ as $\{f_{k}\}$ and let $f_{n} = f_{k_{n}}$ be any subsequence.
Suppose
\begin{equation*}
\limsup_{n \rightarrow \infty} \frac{ \max_{i} \lVert f_{n}^{(i)} \rVert_{1} }{ \var(f_{n}) \frac{\ln n}{n} } < \frac{1}{ M_{0}} \frac{\tanh\left( \frac{\alpha_{0}}{2} \right) }{ \alpha_{0} - \ln \rho_{2}(\alpha_{0})^{2} }.
\end{equation*}
Then, by Lemma \ref{lemma_limsup_properties}, there exists $0 < \varepsilon < \tanh \left( \frac{ \alpha_{0}}{2} \right)$ and $N_{1} \in \mathbb{N}$ such that for all $n \ge N_{1}$ and $1 \le i \le n$
\begin{equation*}
\lVert f_{n}^{(i)} \rVert_{1} \le \frac{1}{M_{0} + \varepsilon} \frac{ \tanh\left( \frac{\alpha_{0}}{2} \right) - \varepsilon }{ \alpha_{0} - \ln \rho_{2}(\alpha_{0})^{2}} (1- \varepsilon) \var(f_{n}) \frac{\ln n }{n},
\end{equation*}
for any such $\varepsilon$ there exists $N_{2} \ge N_{1}$ such that for all $n \ge N_{2}$ and $1 \le i \le n$ such that $f_{n}^{(i)} \ne 0$
\begin{equation*}
\frac{ \lVert f_{n}^{(i)} \rVert_{2}^{2}}{ \lVert f_{n}^{(i)} \rVert_{1}}
\le M_{0} + \varepsilon
\end{equation*}
and there exists $N_{3} \ge N_{2}$ such that for all $n \ge N_{3}$
\begin{equation*}
a_{n} b > 1
\qquad
\frac{ \ln a_{n} b }{a_{n}b -1 } < \alpha_{0} - \ln \rho_{2}(\alpha_{0})^{2},
\end{equation*}
where
\begin{equation*}
a_{n} = \frac{ \tanh\left( \frac{\alpha_{0}}{2} \right) - \varepsilon }{ \alpha_{0} - \ln \rho_{2}(\alpha_{0})^{2}} (1- \varepsilon) \var(f_{n}) \ln n
\qquad
b = \frac{1}{ (1-\varepsilon) \var(f_{n})}.
\end{equation*}
Fix an $n \ge N_{3}$, $f = f_{n}$ and $f_{i} = f_{n}^{(i)}$.
Let $M = M_{0} + \varepsilon$ and
\begin{equation*}
I = \sum_{i=1}^{n} \lVert f_{i} \rVert_{2}^{2}.
\end{equation*}

By Lemma \ref{lemma_expectation_spectrum} and Parseval's theorem
\begin{equation*}
\lVert f_{i} \rVert_{2}^{2}
= \sum_{ i \ni S } \hat{f}(S)^{2}
\qquad
\Longrightarrow
\qquad
I = \sum_{S \subseteq [n] } \lvert S \rvert \hat{f}(S)^{2}.
\end{equation*}
By assumption $\lVert f_{i} \rVert_{2}^{2} \le \frac{a_{n}}{n}$ for all $i$ implying $I \le a_{n}$.
Therefore, by letting
\begin{equation*}
\nu(\mathcal{E}) = \sum_{S \subseteq \mathcal{E} } \hat{f}(S)^{2},
\end{equation*}
we see that
\begin{equation*}
\mathcal{A} = \{ S : \lvert S \rvert > a_{n} b \}
\qquad
\Longrightarrow
\qquad
\nu(\mathcal{A}) \le \frac{1}{b}.
\end{equation*}
Let $\mathcal{B} = \{S : 0 < \lvert S \rvert \le a_{n} b \}$.
Then
\begin{equation*}
\var(f) = \nu(\mathcal{B}) + \nu(\mathcal{A})
\qquad
\Longrightarrow
\qquad
\nu(\mathcal{B}) \ge  \var(f) - \frac{1}{b}.
\end{equation*}
Thus, by definition of $b$,
\begin{equation*}
\varepsilon \var(f)
= \var(f) - \frac{1}{b}
\le \nu(\mathcal{B}).
\end{equation*}

For all $0 \le \delta \le 1$, by Corollary \ref{corollary_p-hypercontractivity},
\begin{equation*}
\sum_{i \ni S} \left( \rho_{1}(\delta) \delta \right)^{ 2 \lvert S \rvert} \hat{f}(S)^{2}
= \lVert T_{ \rho_{1}(\delta) \delta} f_{i} \rVert_{2}^{2}
\le \lVert f_{i} \rVert_{1 + \delta^{2}}^{2}.
\end{equation*}
Moreover, for all $0 \le \gamma \le 1$ and $f_{i} \ne 0$, we have
\begin{align*}
\lVert f_{i} \rVert_{1 + \gamma}^{1+\gamma}
&= \sum_{x} \lvert f_{i} \rvert^{1 + \gamma} \mu(x) \\
&= \lVert f_{i} \rVert_{1} \sum_{x} \lvert f_{i} \rvert^{\gamma} \frac{ \lvert f_{i} \rvert }{ \lVert f_{i} \rVert_{1} } \mu(x) \\
&\le^{(a)} \lVert f_{i} \rVert_{1} \left( \sum_{x} \lvert f_{i} \rvert \frac{ \lvert f_{i} \rvert }{ \lVert f_{i} \rVert_{1} } \mu(x) \right)^{\gamma}\\
&= \lVert f_{i} \rVert_{1} \left( \frac{ \lVert f_{i} \rVert_{2}^{2}}{ \lVert f_{i} \rVert_{1}} \right)^{\gamma} \\
&\le \lVert f_{i} \rVert_{1} M^{\gamma},
\end{align*}
where $(a)$ is an application of Jensen's inequality.
Taking this to the $\frac{2}{1+\gamma}$ power gives
\begin{equation*}
\lVert f_{i} \rVert_{1+\gamma}^{2}
\le \lVert f_{i} \rVert_{1}^{\frac{2}{1+\gamma}} M^{\frac{2\gamma}{1+\gamma}},
\end{equation*}
where this bound holds for all $i$.
Thus, letting $\gamma = \delta^{2}$,
\begin{equation*}
\sum_{i \ni S} \left( \rho_{1}(\delta) \delta \right)^{2 \lvert S \rvert} \hat{f}(S)^{2}
\le \lVert f_{i} \rVert_{1}^{\frac{2}{1 + \delta^{2}}} M^{\frac{2 \delta^{2} }{1+\delta^{2}}}
\end{equation*}
and, letting $\delta^{2} = e^{-\alpha}$,
\begin{align*}
\sum_{S} \lvert S \rvert e^{- (\alpha - \ln \rho_{2}(\alpha)^{2}) \lvert S \rvert} \hat{f}(S)^{2}
&= \sum_{i = 1}^{n} \sum_{i \ni S} e^{- (\alpha - \ln \rho_{2}(\alpha)^{2}) \lvert S \rvert} \hat{f}(S)^{2} \\
&\le \sum_{i =1}^{n} \lVert f_{i} \rVert_{1}^{\frac{2}{1+e^{-\alpha }}} M^{\frac{2 e^{-\alpha}}{1+e^{-\alpha}}}\\
&\le n \left( \frac{1}{M} \frac{a_{n}}{n} \right)^{\frac{2}{1+e^{-\alpha }}} M^{\frac{ 2e^{-\alpha}}{1+e^{-\alpha}}}.
\end{align*}
Combining terms, with $\tanh\left( \frac{\alpha}{2} \right) = \frac{ 1 - e^{-\alpha}}{1+e^{-\alpha}}$, for all $\alpha \ge 0$,
\begin{equation*}
\sum_{S} \lvert S \rvert e^{- (\alpha - \ln \rho_{2}(\alpha)^{2}) \lvert S \rvert} \hat{f}(S)^{2}
\le a_{n}^{ 1 + \tanh\left( \frac{\alpha}{2} \right) } n^{- \tanh \left( \frac{ \alpha}{2} \right)} M^{-2 \tanh\left( \frac{\alpha}{2} \right)}.
\end{equation*}

Then, as $xe^{-\beta x}$ is increasing then decreasing in $x$, for $x,\beta \ge 0$,
\begin{align*}
\sum_{S} \lvert S \rvert e^{- (\alpha - \ln \rho_{2}(\alpha)^{2}) \lvert S \rvert} \hat{f}(S)^{2}
&\ge \sum_{S \in \mathcal{B}} \lvert S \rvert e^{- (\alpha - \ln \rho_{2}(\alpha)^{2}) \lvert S \rvert} \hat{f}(S)^{2} \\
&\ge \min \{e^{-(\alpha - \ln \rho_{2}(\alpha)^{2})},a_{n}be^{-(\alpha - \ln \rho_{2}(\alpha)^{2}) a_{n}b} \} \nu(\mathcal{B}),
\end{align*}
where, for $x > 1$,
\begin{equation*}
\min \{ e^{-\beta}, x e^{-\beta x} \}
=
\begin{cases}
e^{-\beta} & \beta \le \frac{ \ln x}{x -1} \\
xe^{-\beta x} & \beta > \frac{ \ln x}{x -1}
\end{cases}.
\end{equation*}
Therefore, letting $\alpha = \alpha_{0}$, by the conditions imposed on $a_{n},b$ with regard to $\alpha_{0}$,
\begin{equation*}
\min \{e^{-(\alpha_{0} - \ln \rho_{2}(\alpha_{0})^{2})},a_{n}be^{-(\alpha_{0} - \ln \rho_{2}(\alpha_{0})^{2}) a_{n}b} \}
=a_{n}be^{-(\alpha_{0} - \ln \rho_{2}(\alpha_{0})^{2}) a_{n}b}.
\end{equation*}
Combining with the bounds of previous paragraphs
\begin{align*}
\varepsilon \var(f)
&\le \nu(\mathcal{B}) \\
&\le \frac{1}{a_{n} b} e^{(\alpha_{0} - \ln \rho_{2}(\alpha_{0})^{2}) a_{n}b} \sum_{S} \lvert S \rvert e^{- (\alpha_{0} - \ln \rho_{2}(\alpha_{0})^{2}) \lvert S \rvert} \hat{f}(S)^{2} \\
&\le \frac{1}{a_{n} b} e^{(\alpha_{0} - \ln \rho_{2}(\alpha_{0})^{2}) a_{n}b} a_{n}^{ 1 + \tanh\left( \frac{\alpha_{0}}{2} \right) } n^{- \tanh \left( \frac{ \alpha_{0}}{2} \right)} M^{-2 \tanh\left( \frac{\alpha_{0}}{2} \right)}.
\end{align*}
Thus, multiplying both sides by $b$ and taking a $\ln$,
\begin{align*}
\ln \frac{\varepsilon}{1-\varepsilon}
&\le (\alpha_{0} - \ln \rho_{2}(\alpha_{0})^{2}) a_{n}b + \tanh \left( \frac{\alpha_{0}}{2} \right) \ln a_{n} - \tanh \left( \frac{ \alpha_{0}}{2} \right)  \ln n \\
&\quad - 2\tanh\left( \frac{ \alpha_{0}}{2} \right) \ln M,
\end{align*}
where, by definition of $a_{n}$ and $b$,
\begin{equation*}
(\alpha_{0} - \ln \rho_{2}(\alpha_{0})^{2}) a_{n}b
= \left( \tanh\left( \frac{\alpha_{0}}{2} \right) - \varepsilon \right) \ln n
\end{equation*}
and
\begin{equation*}
\ln a_{n}
=  \ln \left(\frac{ \tanh\left( \frac{\alpha_{0}}{2} \right) - \varepsilon }{ \alpha_{0} - \ln \rho_{2}(\alpha_{0})^{2}} (1- \varepsilon) \var(f) \ln n \right)
= \ln \ln n + \ln o(n^{\varepsilon}) + O(1),
\end{equation*}
where $\var(f) = o(n^{\varepsilon})$ by assumption.
Combining
\begin{equation*}
\ln \frac{\varepsilon}{1-\varepsilon}
\le - \varepsilon \ln n + \ln \ln n + \ln o(n^{\varepsilon}) + O(1).
\end{equation*}
Hence
\begin{equation*}
\frac{ \varepsilon }{1- \varepsilon}
= o(1),
\end{equation*}
a contradiction.
Thus
\begin{equation*}
\limsup_{n \rightarrow \infty} \frac{ \max_{i} \lVert f_{n}^{(i)} \rVert_{1} }{ \var(f_{n}) \frac{\ln n}{n} } \ge \frac{1}{ M_{0}} \frac{\tanh\left( \frac{\alpha_{0}}{2} \right) }{ \alpha_{0} - \ln \rho_{2}(\alpha_{0})^{2} }
\end{equation*}
and, as $\{f_{n}\}$ was an arbitrary subsequence, this extends to the $\liminf$ by Lemma \ref{lemma_limsup_properties}.
As $\alpha_{0}$ was arbitrary, the result follows.
\qed

\begin{lemma}
\label{lemma_limsup_properties}
For any sequence $\{x_{n} \in \mathbb{R}\}$, the following statements hold.
\begin{enumerate}[i)]
\item
For all $\varepsilon >0$, there exists $N \in \mathbb{N}$ such that $n \ge N$ implies $x_{n} \le \limsup_{n \rightarrow \infty} x_{n} + \varepsilon$;
\item
If $\limsup_{n \rightarrow \infty} x_{n} < c$, then for any decreasing continuous function $f : (0,b] \rightarrow [0,\infty)$ such that $f(b) = 0$ and $\lim_{x \downarrow 0} f(x) = c$  there exists $0< \varepsilon < b$ and $N \in \mathbb{N}$ such that $n \ge N$ implies $x_{n} \le f(\varepsilon)$;
\item
If, for all subsequences $\{x_{n_{k}}\}$, $\limsup_{k \rightarrow \infty} x_{n_{k}} \ge c$, then $\liminf_{n \rightarrow \infty} x_{n} \ge c$.
\end{enumerate}
\end{lemma}

\begin{proof}
(ii)
If $c = \infty$, as the inequality is strict, there is a finite positive $c_{1}$ that satisfies the inequality.
Moreover, $f(a) = c_{1}$ for some $0 < a < b$ by decreasing and continuity, and $g(x) = f(x+a)$ on $(0,b-a]$ satisfies the conditions.
Thus WLOG suppose $c$ is finite.
There exists $\delta > 0$ such that $\limsup_{n \rightarrow \infty} x_{n} \le c - \delta$.
Moreover, for all $0 < \gamma < \delta$ there exists $N \in \mathbb{N}$ such that $n \ge N$ implies $x_{n} \le c - \delta + \gamma < c$.
Choose by decreasing and continuity $\varepsilon$ such that $c - \delta + \gamma \le f(\varepsilon)$.
(iii)
Let $L = \liminf_{n \rightarrow \infty} x_{n}$ and $\varepsilon > 0$.
Then there exists a subsequence $\{x_{n_{k}}\}$ such that
\begin{equation*}
L + \varepsilon > x_{n_{k}}
\end{equation*}
for all $k$.
Thus
\begin{equation*}
L + \varepsilon \ge \limsup_{k \rightarrow \infty} x_{n_{k}} \ge c.
\end{equation*}
\end{proof}

\section{Large derivatives}

We use a result of Rossignol to relate the derivative of the expectation to the sum of $L^1$ influences.

\begin{lemma}{\rm \cite{Rossignol2006} }
Let $\mu$ be the $p$-biased measure and $f : \{-1,1\}^{n} \rightarrow \mathbb{R}$, then
\begin{equation*}
\frac{d}{dp} \ex{f}
= \sum_{i=1}^{n} E \left[ f \circ \tau_{i}^{+} - f \circ \tau_{i}^{-} \right].
\end{equation*}
\end{lemma}

\begin{proof}
Let
\begin{equation*}
\mu(x) = \prod_{i=1}^{n} \mu_{i}(x),
\end{equation*}
where
\begin{equation*}
\mu_{i}(x) = \frac{2p-1}{2}(x_{i}+1) + 1 - p.
\end{equation*}
Then
\begin{equation*}
\frac{d}{dp} \mu(x) 
= \sum_{i=1}^{n} \left( \frac{d}{dp} \mu_{i}(x) \right) \prod_{j \ne i } \mu_{j}(x) 
= \sum_{i=1}^{n} x_{i} \prod_{j \ne i } \mu_{j}(x).
\end{equation*}
\end{proof}

For any $n$, an element $\pi$ of the symmetric group $S_{n}$ acts on $x \in \{-1,1\}^{n}$ by $[\pi(x)]_{i} = x_{\pi(i)}$.
Given a real-valued boolean function $f$ its \emph{symmetry group} $\mathcal{G}$ is the following subset of $S_{n}$
\begin{equation*}
\mathcal{G} := \{ \pi \in S_{n} : f \circ \pi = f \}.
\end{equation*}
The function $f$ is said to be \emph{symmetric} if $\mathcal{G}$ is transitive, i.e. for all $i,j \in \{1,\ldots, n\}$ there exists $\pi \in \mathcal{G}$ such that $\pi(i) = j$.
The boolean domain is endowed with the standard partial order, $x \preceq y$ if $x_{i} \le y_{i}$ for all $i$.
We call a function \emph{monotone} (or more precisely monotone increasing) if $x \preceq y$ implies $f(x) \le f(y)$.

Combining with Lemma \ref{lemma_expectation_spectrum}, for any monotone function
\begin{equation*}
\frac{d}{dp} \ex{f}
\ge 2\sum_{i=1}^{n} \lVert f - \exi{f} \rVert_{1}.
\end{equation*}
If $f$ is symmetric then the sum on the RHS is constant in $i$, i.e. given $i$ and $j$ choose $\pi \in \mathcal{G}$ with $\pi(i) = j$ then
\begin{equation*}
f \circ \tau_{i}^{\pm}
= (f \circ \pi) \circ \tau_{i}^{\pm}
= (f \circ \tau_{j}^{\pm}) \circ \pi
= f \circ \tau_{j}^{\pm}.
\end{equation*}
Therefore, for a sequence of non-constant monotone symmetric functions all bounded by $b$, by Equation \ref{equation_real_kkl_bounded}, eventually
\begin{equation*}
\frac{d}{dp} \ex{f_{n}}
\ge \frac{9}{10b} \frac{1}{1 + \left \lvert \ln \frac{p}{1-p} \right \rvert } \var(f_{n}) \ln n.
\end{equation*}

\subsection{Conditions implying monotonicity}

Let $\mu = \mu_{1} \otimes \ldots \otimes \mu_{n}$ where each $\mu_{i}$ is the measure with $\exi{x_{i}} = 2p_{i} - 1$ and
\begin{equation*}
g(\underline{p}) = \ex{f}
\end{equation*}
for some $f : \{-1,1\}^{n} \rightarrow \mathbb{R}$.

For $n = 1$
\begin{equation*}
g(p) = (1 - p) f(-1) + p f(1)
\qquad
g'(p) = f(1) - f(-1).
\end{equation*}
Thus $g \ge 0$ if and only if $f \ge 0$ and $g' \ge 0$ if and only if $f$ is monotone.

For the general case $g \ge 0 $ if and only if $f \ge 0$, evaluate $g$ at e.g. $\underline{1}$ for $f(\underline{1})$.
Moreover,
\begin{equation*}
\exi{f}
= (1-p_{i}) f \circ \tau_{i}^{-} + p_{i} f \circ \tau_{i}^{+}
= f \circ \tau_{i}^{-} + p_{i}( f \circ \tau_{i}^{+} - f \circ \tau_{i}^{-}),
\end{equation*}
where $f \circ \tau_{i}^{\pm}$ are independent of $x_{i}$.
Thus
\begin{equation*}
\ex{f} = \ex{f \circ \tau_{i}^{-} } + p_{i} \ex{ f \circ \tau_{i}^{+} - f \circ \tau_{i}^{-} }
\end{equation*}
and
\begin{equation*}
\frac{ \partial}{ \partial_{i} p_{i}} g(\underline{p})
= \ex{f \circ \tau_{i}^{+} - f \circ \tau_{i}^{-}}.
\end{equation*}
Hence $f$ is monotone if and only if, for all $i$, $\frac{ \partial}{\partial_{i}} g(\underline{p}) \ge 0$.

\subsection{Weak conditions}

The combination of full monotonicity and symmetry is a rather strong condition.
It suffices to show that there exists $\gamma > 0$ such that eventually
\begin{equation*}
\sum_{i=1}^{n} \ex{ f_{n} \circ \tau_{i}^{+} - f_{n} \circ \tau_{i}^{-} }
\ge \gamma n \max_{i} \lVert f_{n} \circ \tau_{i}^{+} - f_{n} \circ \tau_{i}^{-} \rVert_{1}.
\end{equation*}

This can be decomposed into a weak monotonicity and weak symmetry condition as follows.
A real-valued boolean function $f : \{-1,1\}^{n} \rightarrow \mathbb{R}$ is \emph{weakly monotone} if there exists $\alpha > 0$ such that
\begin{equation*}
\sum_{i=1}^{n} \ex{f \circ \tau_{i}^{+} - f \circ \tau_{i}^{-} }
\ge \alpha \sum_{i=1}^{n} \lVert f \circ \tau_{i}^{+} - f \circ \tau_{i}^{-} \rVert_{1}.
\end{equation*}
Similarly, a real-valued boolean function $f : \{-1,1\}^{n} \rightarrow \mathbb{R}$ is \emph{weakly symmetric} if there exists $\beta > 0$ such that
\begin{equation*}
\sum_{i=1}^{n} \lVert f \circ \tau_{i}^{+} - f \circ \tau_{i}^{-} \rVert_{1}
\ge \beta n \max_{i} \lVert f \circ \tau_{i}^{+} - f \circ \tau_{i}^{-} \rVert_{1}.
\end{equation*}
A sequence of functions $\{f_{n}\}$ is \emph{weakly monotone (symmetric)} if $f_{n}$ is eventually weakly monotone (symmetric) for some fixed $\alpha$ ($\beta$).
It should be noted that $\ex{\cdot}$, $\lVert \cdot \lVert_{1}$ are implicit functions of $p$.

\section{Tribes}

Recall the tribes \cite{Ben-Or1990} boolean function defined as the logical OR of a collection of disjoint logical ANDs.
This function is known to demonstrate the tightness of the original KKL.
As Theorem \ref{theorem_real_kkl} includes boolean functions, this applies similarly, and we provided the explicit calculations for completeness.

Consider the tribes boolean function with $\frac{n}{\ell}$ equal-sized tribes of size $\ell$.
Then, the influences satisfy
\begin{equation*}
I_{i} = 2^{-(\ell - 1) }\left(1- 2^{-\ell} \right)^{\frac{n}{\ell}-1}
\end{equation*}
for all $i$ with variance
\begin{equation*}
4\left(1- 2^{-\ell} \right)^{\frac{n}{\ell}}\left(1- \left(1- 2^{-\ell} \right)^{\frac{n}{\ell}} \right).
\end{equation*}
Thus, the ratio of influence to variance is
\begin{equation*}
2^{-(\ell -1)} \frac{1}{ 4 \left(1- 2^{-\ell} \right) \left(1- \left(1- 2^{-\ell} \right)^{\frac{n}{\ell}} \right) }.
\end{equation*}

Let $n = m 2^{m}$ and $f_{n}$ be the boolean function defined by uniform tribes of size $\ell = m$.
Then
\begin{equation*}
I_{i}(f_{n})
= 2^{-(m - 1) }\left(1- \frac{1}{2^{m}} \right)^{2^{m}-1},
\end{equation*}
where
\begin{equation*}
\left(1- \frac{1}{2^{m}} \right)^{2^{m}}
\rightarrow \frac{1}{e}
\end{equation*}
and
\begin{equation*}
\frac{ 2^{-m}}{\frac{ \log n }{n}}
=\frac{ 2^{-m}}{\frac{ m + \log m }{ m 2^{m}}}
= \frac{m}{m+\log m }
= 1 + o(1).
\end{equation*}
Thus, we have
\begin{align*}
\liminf_{n \rightarrow \infty} \frac{I_{i}(f_{n})}{ \var(f_{n}) \frac{ \ln n }{n} }
&\le \lim_{m \rightarrow \infty} \frac{ 2 \log e}{ 4 \left( 1- \frac{1}{e} \right) + o(1) }  (1+ o(1)) \\
&= \frac{1}{2} \frac{ \log e }{ 1 - \frac{1}{e} }.
\end{align*}

Let $n = m2^{m+k}$, for a fixed $k$, and $f_{n}$ be the boolean function defined by uniform tribes of size $\ell = m$.
Then
\begin{equation*}
I_{i}(f_{n})
= 2^{-(m - 1) }\left(1- \frac{1}{2^{m}} \right)^{2^{m+k}-1},
\end{equation*}
where
\begin{equation*}
\left(1- \frac{1}{2^{m}} \right)^{2^{m+k}}
\rightarrow e^{-2^{k}}
\end{equation*}
and
\begin{equation*}
\frac{ 2^{-m}}{\frac{ \log n }{n}}
=\frac{ 2^{-m}}{\frac{ m+k + \log m }{ m 2^{m + k }}}
= 2^{k} \frac{m}{m+k+\log m }
= 2^{k}(1 + o(1)).
\end{equation*}
Thus
\begin{equation*}
\liminf_{n \rightarrow \infty} \frac{I_{i}(f_{n})}{ \var(f_{n}) \frac{ \ln n }{n} }
\le \frac{1}{2} 2^{k} \frac{ \log e }{ 1 - e^{-2^{k}}}.
\end{equation*}
Letting $k$ tend to $-\infty$, for all $\varepsilon > 0$, there exists $f_{n} : \{-1,1\}^{n} \rightarrow \{-1,1\}$ such that
\begin{equation*}
\liminf_{n \rightarrow \infty} \frac{I_{i}(f_{n})}{ \var(f_{n}) \frac{ \ln n }{n} }
\le \frac{1}{2} \log e + \varepsilon.
\end{equation*}


\begin{thebibliography}{10}

\bibitem{KKL1988}
J.~Kahn, G.~Kalai, and N.~Linial, ``The influence of variables on boolean
  functions,'' {\em Proceedings of 29th Annual IEEE Symposium on Foundations of
  Computer Science}, pp.~330--341, 1988.

\bibitem{Bourgain1992}
J.~Bourgain, J.~Kahn, G.~Kalai, Y.~Katznelson, and N.~Linial, ``The influence
  of variables in product spaces,'' {\em Israel Journal of Mathematics},
  vol.~77, pp.~55--64, 1992.

\bibitem{Hatami2009}
H.~Hatami, ``Decision trees and influences of variables over product
  probability spaces,'' {\em Combinatorics, Probability and Computing},
  vol.~18, pp.~357--360, 2009.

\bibitem{Keller2011}
N.~Keller, ``On the influence of varibales on boolean functions in product
  spaces,'' {\em Combinatorics, Probability and Computing}, vol.~20,
  pp.~83--102, 2011.

\bibitem{Cordero-Erausquin2012}
D.~Cordero-Erausquin and M.~Ledoux, {\em Hypercontractive Measures, Talagrand's
  Inequality, and Influences}, pp.~169--189.
\newblock Berlin, Heidelberg: Springer Berlin Heidelberg, 2012.

\bibitem{Keller2012}
N.~Keller, E.~Mossel, and A.~Sen, ``Geometric influences,'' {\em Annals of
  Probability}, vol.~40, no.~3, pp.~1335--1166, 2012.

\bibitem{ODonnell2013}
R.~O'Donnell and K.~Wimmer, ``Sharpness of {KKL} on {S}chreier graphs,'' {\em
  Electronic Communications in Probability}, vol.~18, no.~18, pp.~1--12, 2013.

\bibitem{Talagrand1994}
M.~Talagrand, ``On {R}usso's approximate zero-one law,'' {\em The Annals of
  Probability}, pp.~1576--1587, 1994.

\bibitem{Friedgut1996}
E.~Friedgut and G.~Kalai, ``Every monotone graph property has a sharp
  threshold,'' {\em Proceedings of the American Mathematical Society},
  pp.~2993--3002, 1996.

\bibitem{Rossignol2006}
R.~Rossignol, ``Threshold for monotone symmetric properties through a
  logarithmic {S}obolev inequality,'' {\em The Annals of Probability}, vol.~34,
  pp.~1707--1725, 2006.

\bibitem{Margulis1974}
G.~A. Margulis, ``Probabilistic characteristics of graphs with large
  connectivity,'' {\em Probability Theory and Related Fields}, vol.~10, no.~2,
  pp.~101--108, 1974.

\bibitem{Russo1982}
L.~Russo, ``An approximate zero-one law,'' {\em Probability Theory and Related
  Fields}, vol.~61, pp.~129--139, 1982.

\bibitem{Kelman2021}
E.~Kelman, S.~Khot, G.~Kindler, D.~Minzer, and M.~Safra, ``Theorems of {KKL},
  {F}lriedgut, and {T}alagrand via random restrictions and log-{S}obolev
  inequality,'' {\em 12th Innovations in Theoretical Computer Science
  Conference (ITCS 2021)}, pp.~26:1--26:17, 2021.

\bibitem{Ben-Or1990}
M.~Ben-Or and N.~Linial, ``Collective coin flipping,'' {\em Randomness and
  Computation}, vol.~5, pp.~91--115, 1990.

\bibitem{odonnell14}
R.~O’Donnell, {\em Analysis of Boolean Functions}.
\newblock Cambridge University Press, 2014.

\bibitem{Latala2000}
R.~Latala and K.~Oleszkiewicz, ``Between {S}obolev and {P}oincar\'{e},'' {\em
  Geometric aspects of functional analysis}, pp.~147--168, 2000.

\bibitem{Beckner1975}
W.~Beckner, ``Inequalities in {F}ourier analysis,'' {\em Annals of
  Mathematics}, vol.~102, pp.~159--182, 1975.

\end{thebibliography}
\end{document}